\documentclass[12pt,letterpaper]{article}
\usepackage{amsthm,amsmath,amssymb}
\usepackage{natbib,url,enumerate,xcolor}
\usepackage{graphicx}
\usepackage[top=2cm,bottom=2cm,left=2cm,right=2cm]{geometry}
\usepackage[colorlinks=true,citecolor=black!70!blue,urlcolor=blue,linkcolor=blue]{hyperref}
\newtheorem{thm}{Theorem}
 
\newtheorem{lemma}[thm]{Lemma}
\newtheorem{cor}[thm]{Corollary}
\theoremstyle{definition}

\theoremstyle{remark}

\newcommand*{\set}[1]{\left\{#1\right\}} 
\newcommand*{\infe}{\underline{\tilde{\epsilon}}}
\newcommand*{\supe}{\overline{\tilde{\epsilon}}}
\newcommand*{\peake}{\epsilon^*}
\newcommand*{\infx}{a}
\newcommand*{\supx}{b}

\begin{document}
\title{Reversals of signal-posterior monotonicity imply a bias of screening} 
\author{
Sander Heinsalu
\thanks{Research School of Economics, Australian National University. HW Arndt Building, 25a Kingsley St, Acton ACT 2601, Australia.
Email: sander.heinsalu@anu.edu.au, 
website: \url{https://sanderheinsalu.com/}. 
}
}
\date{\today}
\maketitle

\begin{abstract}
This note strengthens the main result of \cite{lagziel+lehrer2019} (LL) ``A bias in screening'' using \cite{chambers+healy2011} (CH) ``Reversals of signal-posterior monotonicity for any bounded prior''. 
LL show that the conditional expectation of an unobserved variable of interest, given that a noisy signal of it exceeds a cutoff, may decrease in the cutoff. 
CH prove that the distribution of a variable given a lower signal may first order stochastically dominate the distribution given a higher signal. 

The nonmonotonicity result is also extended to the empirically relevant exponential and Pareto distributions, and a wide range of signals. 



\end{abstract}

\section{Introduction}

\cite{lagziel+lehrer2019} (LL) prove that for any bounded random variable of interest, there exists a noisy signal such that the expectation of the variable conditional the signal exceeding a lower cutoff is larger than conditional on the signal passing a higher cutoff. 
\cite{chambers+healy2011} (CH) show that for any bounded support, there exists a signal such that for any random variable on that support, the distribution of the variable conditional on a lower signal realization first order stochastically dominates (FOSDs) the distribution conditional on a higher realization. 
This note strengthens the result of LL, both by replacing the conditional expectation order with the FOSD order, and by re-ordering the quantifiers. The proof modifies that of CH and introduces the novel technique of defining the cdf of one random variable as (one minus the pdf) of another. 

The present work also extends the results to exponential and Pareto distributions, which are empirically relevant for studying income and wealth. CH did not consider such unbounded distributions, and LL only prove a weaker result than their main theorem for the unbounded case (Lemma 2 in LL). 
For exponential and thicker-tailed distributions, Corollary~\ref{cor:independent} below shows that most pairs of signals produce the counterintuitive FOSD ranking, in contrast to the proofs of CH and LL, which rely on two specific signal realizations. 

The first substantive section establishes the connection between LL and CH. After that, Section~\ref{sec:intuitive} studies the limitations of the results, providing a sufficient condition to rule out the counterintuitive ranking of conditional distributions. Section~\ref{sec:exp} extends the counterintuitive FOSD ordering to unbounded distributions, including exponential and Pareto.

\section{The connection between the previous papers}
\label{sec:connection}

Before stating the results, the notation is introduced in Table~\ref{tab:notation}. After that, the central inequality 
in the first theorem in LL is generalized as Theorem~\ref{thm:extremescreeningold} below. The proof is modeled on the main result of CH, but introduces the trick of constructing a cdf as one minus the pdf of another random variable. 

\begin{table}[h!]
\centering
\caption{Notation correspondence. 
}
	\label{tab:notation}
\begin{tabular}{ccc}
\hline
This paper & CH & LL \\
\hline
 $X$ & $X$	& $V$ \\ 
 $\tilde{\epsilon}$ & $\tilde{\epsilon}$ & $N$ \\
 $Z$ & $Z$	& $V+N$ \\ 
 $F_X$ & $F$ & -- \\
 $F_{Z|X}(z|x)$ & $G_x(z)$ & -- \\
 $[a,b]$ & $[a,b]$ & $[0,1]$ or $[\underline{V},\overline{V}]$ \\
 $\mathbb{E}[X|Z\geq b_i]$ & -- & $\pi(b_i)$ \\
\hline
\end{tabular} 
\end{table}



\begin{thm} 
\label{thm:extremescreeningold}
Fix any $a<b$. There exists a family of conditional signal distributions $\set{F_{Z|X}(\cdot|x)}$ such that for any $X$ with support $[a,b]$, $F_{X|Z}(\cdot|b)$ strictly first-order stochastically dominates $F_{X|Z}(\cdot|2b-a)$. 
Furthermore, $\set{F_{Z|X}(\cdot|x)}$ forms an independent additive signal, and each $F_{Z|X}(\cdot|x)$ is unimodal. 
There exists a signal $S$ such that $F_{X|S}(\cdot|S\geq z) =F_{X|Z}(\cdot|z)$. 
\end{thm}
\begin{proof}
W.l.o.g., let $a=0,b=1$. Let the signal pdf be the union of a right-angled triangle and a rectangle: $f_{Z|X}(z|x):=\begin{cases}
1-\frac{2}{3}(z-x) & \text{if } z\in[x,x+1),\\
\frac{1}{3} & \text{if } z\in[x+1,x+2).\\
\end{cases}$
The noise term $\tilde{\epsilon}=Z-X$ is additive and independent. 
The positive mean of $\tilde{\epsilon}\geq0$ simplifies formulas and is w.l.o.g., because Bayes' rule de-biases the signal. 

Conditional on $z''=2 =2b-a$, the posterior equals the prior: $F_{X|Z}(x|z'')=F_{X}(x)$ for all $x$. 
Conditional on $z'=1 =b$, integrating by parts yields the posterior 
\begin{align}
\label{Fxzw}
&F_{X|Z}(w|z') =\frac{\int_{0}^{w}[x-z'+\frac{3}{2}]dF_{X}(x)}{\int_{0}^{1}[x-z'+\frac{3}{2}]dF_{X}(x)}
\notag =\frac{[\frac{3}{2}-z']F_{X}(w)+wF_{X}(w)-\int_{0}^{w}F_{X}(x)dx}{\frac{3}{2}-z'+1-\int_{0}^{1}F_{X}(x)dx}
\\& =\frac{[\frac{3}{2}-z']F_{X}(w)+\int_{0}^{w}[F_{X}(w)-F_{X}(x)]dx}{\frac{3}{2}-z'+\int_{0}^{1}[1-F_{X}(x)]dx}
\end{align}
Clearly $F_{X|Z}(0|z')=0\leq F_{X}(0)$ and $F_{X|Z}(1|z')=1\leq F_{X}(1)$ for any $z'$. 
Consider $w\in(0,1)$ next. 

The numerator of~(\ref{Fxzw}) is  
\begin{align*}
&F_{X}(w)\left\{\frac{3}{2}-z'+\int_{0}^{w}\left[1-\frac{F_{X}(x)}{F_{X}(w)}\right]dx\right\} 
< F_{X}(w)\left\{\frac{3}{2}-z'+\int_{0}^{w}\left[1-F_{X}(x)\right]dx\right\}
\\&\leq F_{X}(w)\left\{\frac{3}{2}-z'+\int_{0}^{1}\left[1-F_{X}(x)\right]dx\right\} . 
\end{align*}
The term in braces is the denominator of $F_{X|Z}(w|z')$, thus dividing by it yields  $F_{X|Z}(w|z') 
<F_{X}(w) =F_{X|Z}(w|z'')$. Therefore $F_{X|Z}(w|z')$ FOSDs $F_{X|Z}(w|z'')$ despite $z'<z''$. 

Construct another signal $S$ from $Z$ by taking $F_{S|X}(\cdot|x) =1-f_{Z|X}(\cdot|x)$, so that $F_{X|S}(\cdot|S\geq z) =F_{X|Z}(\cdot|z)$. For any $x$, $F_{S|X}(\cdot|x)$ is a cdf: increasing, with $F_{S|X}(x|x)=0$ and $F_{S|X}(x+2|x)=1$. 
\end{proof}

To the best of the author's knowledge, the technique in Theorem~\ref{thm:extremescreeningold} of 
constructing a signal $S$ from another signal $Z$ such that the information content of observing $S\geq \supx$ is equivalent to $Z=\supx$ is new. 
The main difficulty in the proof is to ensure that the cdf of $S$, which equals one minus the pdf of $Z$, is increasing and has the maximum value $1$. 

Theorem~\ref{thm:extremescreeningold} implies the following corollary, which is the central inequality in the first theorem of LL.
\begin{cor}
\label{cor:lagzielthm1old}
For every bounded random variable $X$, there exist a noise variable $\tilde{\epsilon}$ and cutoffs $b_1<b_2$ such that $\mathbb{E}[X|X+\tilde{\epsilon}\geq b_1] >\mathbb{E}[X] =\mathbb{E}[X|X+\tilde{\epsilon}\geq b_2]$. 
\end{cor}
\begin{proof}
W.l.o.g., let $\underline{X}=0,\overline{X}=1$. Take $b_1=1$, $b_2=2$ and $\tilde{\epsilon}=S-X$, where the cdf of $S$ is defined in the proof of Theorem~\ref{thm:extremescreeningold}. 
Then $F_{X|S}(\cdot|S\geq 1)<F_{X|S}(\cdot|S\geq 2) =F_{X}(\cdot)$, i.e.\ $F_{X|S}(\cdot|S\geq 1)$ strictly FOSDs $F_{X|S}(\cdot|S\geq 2)$. This implies $\mathbb{E}[X|X+\tilde{\epsilon}\geq 1] >\mathbb{E}[X|X+\tilde{\epsilon}\geq 2] =\mathbb{E}[X]$. 
\end{proof}

Claim 1 in LL says that for every bounded random variable $X$, there exists a continuous noise variable $\tilde{\epsilon}$ and cutoffs $b_1<b_2$ such that $\mathbb{E}[X|X+\tilde{\epsilon}\geq b_1] >\mathbb{E}[X|X+\tilde{\epsilon}\geq b_2] $. 
This claim is strengthened in the following corollary, which is derived analogously to CH. 
\begin{cor}
\label{cor:continuous}
For every bounded $X$, there exist a continuous noise variable $\tilde{\epsilon}$ and cutoffs $b_1<b_2$ such that for $S=X+\tilde{\epsilon}$, and any $w\in(0,1)$, $F_{X|S}(w|S\geq b_1) <F_{X|S}(w|S\geq b_2)$. 
\end{cor}
\begin{proof}
Take $\iota\in(0,1)$ small, $\xi\in[1,\frac{2-\iota}{1+\iota})$ and $h=\frac{2+\iota+\iota^2-\xi}{2+\xi+\iota}$ and define 
\begin{align}
\label{fzxmod}
f_{Z|X,\iota}(z|x):=\begin{cases}
1-\frac{(z-x)(1-h)}{\xi} & \text{if } z\in[x,x+\xi),\\
h-(z-x-\xi)\iota & \text{if } z\in[x+\xi,x+\xi+1),\\
h-\iota -(z-x-\xi-1)\frac{h-\iota}{\iota} & \text{if } z\in[x+\xi+1,x+\xi+1+\iota].\\
\end{cases}
\end{align}
Conditional on $z''=1+\xi$, the posterior is 
\begin{align*}
&F_{X|Z,\iota}(w|z'') =\frac{\int_{0}^{w}\left[h-(z''-x-\xi)\iota\right]dF_{X}(x)}{\int_{0}^{1}\left[h-(z''-y-\xi)\iota\right]dF_{X}(y)}
\notag =\frac{\int_{0}^{w}\left[h/\iota-1+x\right]dF_{X}(x)}{\int_{0}^{1}\left[h/\iota-1+y\right]dF_{X}(y)}, 
\end{align*} 
which strictly FOSDs the prior. 
Clearly $\lim_{\iota\rightarrow0}F_{X|Z,\iota}(w|z'') =F_{X}(w)$ for any $w$. 

Conditional on $z'=1$, the posterior is 
\begin{align*}
&F_{X|Z,\iota}(w|z') =\frac{\int_{0}^{w}\left[1-\frac{(z'-x)(1-h)}{\xi}\right]dF_{X}(x)}{\int_{0}^{1}\left[1-\frac{(z'-y)(1-h)}{\xi}\right]dF_{X}(y)}
\notag =\frac{\int_{0}^{w}\left[\frac{\xi}{1-h}-z'+x\right]dF_{X}(x)}{\int_{0}^{1}\left[\frac{\xi}{1-h}-z'+y\right]dF_{X}(y)}. 
\end{align*} 
The same reasoning as in the proof of Theorem~\ref{thm:extremescreeningold} establishes that $F_{X|Z,\iota}(w|z')<F_{X}(w)$ independently of $\iota$. Therefore there exists $\iota^*>0$ s.t.\ for all $\iota\in[0,\iota^*]$, $F_{X|Z,\iota}(w|z')<F_{X|Z,\iota}(w|z'')$. 

Take $F_{S|X,\iota} =1-f_{Z|X,\iota}$, $b_1=z'=1$ and $b_2=z''=1+\xi$ to complete the proof. 
\end{proof} 

Corollary~\ref{cor:continuous} concludes the generalization of LL using CH. 
The next section delimits the results by providing sufficient conditions for an updated distribution given a lower signal not to FOSD the distribution given a higher signal. After that, the counterintuitive ranking is extended to the unbounded exponential and Pareto distributions.

\section{Ruling out a counterintuitive FOSD order}
\label{sec:intuitive}

The following lemma provides conditions that preclude the counterintuitive FOSD ranking found in Theorem~\ref{thm:extremescreeningold} and Corollary~\ref{cor:continuous}. The intuition for the conditions is that either a lower signal rules out some large values of $X$, which a higher signal permits, or a higher signal rules out small realizations of $X$, which a lower signal allows. 
To state the lemma, denote the range of $\tilde{\epsilon}$ (the bounds of its support) by $[\infe,\supe]$. 

\begin{lemma}
\label{lem:ruleout}
Fix $\tilde{\epsilon}$ and $z_1<z_2$. If either $x\in(z_1-\infe,z_2-\infe]$ or $x\in[z_1-\supe,z_2-\supe)$ has positive probability, then $F_{X|Z}(\cdot|z_1)$ does not FOSD $F_{X|Z}(\cdot|z_2)$ even weakly. 
\end{lemma}
\begin{proof}
If $\Pr\left(x\in(z_1-\infe,z_2-\infe]\right)>0$, then $z_1$ rules out the high values $x\in(z_1-\infe,z_2-\infe]$ of $X$, which $z_2$ allows. This implies $F_{X|Z}(x|z_1) =1 >F_{X|Z}(x|z_2)$. 
If $\Pr\left(x \in[z_1-\supe,z_2-\supe)\right)>0$, then $z_2$ rules out the low realizations $x\in(z_1-\infe,z_2-\infe]$ of $X$, which $z_1$ permits. This implies $F_{X|Z}(x|z_1) >0 =F_{X|Z}(x|z_2)$ for some $x \in(z_1-\infe,z_2-\infe)$. 
\end{proof}
Lemma~\ref{lem:ruleout} only rules out a FOSD ranking non-monotone in specific signals (or non-monotone in the lower cutoff for the transformed signal $S$ used in Theorem~\ref{thm:extremescreeningold}). 
The conditional expectation of $X$ given a signal above a cutoff may still be non-monotone in the cutoff,\footnote{For example, take $f_X(x)=(1-\iota)\text{\textbf{1}}_{[0,1]}+\iota \phi(x)$, where $\phi$ is the standard normal pdf and $\iota>0$ is small enough, and use the signal cdf $F_{S|X}$ and cutoffs $a,2b-a$ from the proof of Theorem~\ref{thm:extremescreeningold}. 
} 
i.e., the LL result may still hold when the stronger CH result fails. 

The next corollary gives sufficient conditions for no pair of signals to generate the reversed FOSD ranking found in Theorem~\ref{thm:extremescreeningold} and Corollary~\ref{cor:continuous}. 
The sufficient conditions hold when $X$ is unbounded and has a positive density, 
but the noise $\tilde{\epsilon}$ is bounded. 
The conditions also hold when both $X$ and $\tilde{\epsilon}$ are bounded below, but not above, or both above, but not below. 

\begin{cor}
\label{cor:ruleout}
If the support $[\infx,\supx]$ of $X$ is an interval and either 
(i) $\supe-\infe\leq \supx-\infx$, 
(ii) $\infe>-\infty$ and $\supx=\infty$, or
(iii) $\supe<\infty$ and $\infx=-\infty$, 
then for any $z_1<z_2$, $F_{X|Z}(\cdot|z_1)$ does not FOSD $F_{X|Z}(\cdot|z_2)$ even weakly.  
\end{cor}
\begin{proof} 
If the support $[\infx,\supx]$ of $X$ is an interval and $\supe-\infe\leq \supx-\infx$, then for any $z_1<z_2$, either $\Pr\left(x\in(z_1-\infe,z_2-\infe]\right)>0$ or $\Pr\left(x \in[z_1-\supe,z_2-\supe)\right)>0$, so Lemma~\ref{lem:ruleout} applies. 

If $\infe>-\infty$, then w.l.o.g.\ take $\infe=0$, because translating the signal $Z$ does not affect Bayesian updating. 
If $\infe=0$, then $z_i$ rules out $x>z_i$, so conditional on $z_1$, $x\in(z_1,z_2]$ is impossible. 
If $\supx=\infty$ and the support of $X$ is an interval, then for any possible $z_2>z_1\geq \infx$, the support of $X$ contains $(z_1,z_2]$. Thus for $\eta>0$ small enough, $F_{X|Z}(z_2-\eta|z_1)>0=F_{X|Z}(z_2-\eta|z_2)$. 

If $\infe=-\infty$, but $\supe<\infty$, then w.l.o.g.\ take $\supe=0$. If $\infx=-\infty$ and $z_1<z_2\leq \supx$, then $(z_1,z_2]$ is in the support of $X$, so for $\eta>0$ small enough, $F_{X|Z}(z_1+\eta|z_1)=1>F_{X|Z}(z_1+\eta|z_2)$. 
\end{proof}  

An open question is whether the reversed FOSD ranking can be generated for some signals when both $X$ and $\tilde{\epsilon}$ are unbounded below and above. 
The next section establishes the reversed ranking for $X$ bounded on one side and $\tilde{\epsilon}$ bounded on the opposite side or unbounded.

\section{Variable of interest bounded on one side}
\label{sec:exp}

The following theorem proves that the counterintuitive ranking of conditional distributions occurs for a wide range of signals when $X$ is bounded on at least one side. Corollaries of this result cover the empirically relevant exponential and Pareto distributions, as discussed subsequently. The theorem allows both independent and dependent $X$ and $\tilde{\epsilon}$. 

\begin{thm}
\label{thm:unbounded}
If $X\leq\supx<\infty$ and $\frac{d\ln f_{Z|X}(z|x)}{dz}$ decreases in $x$, then $\frac{dF_{X|Z}(w|z)}{dz} \geq 0$ for any $w\leq \supx$, with strict inequality if $\frac{d\ln f_{Z|X}(z|\cdot)}{dz}$ is strictly decreasing and $w\in(\infx, \supx)$. 
\end{thm}
\begin{proof}
Recall that $\infx:=\inf X\in[-\infty,\supx)$. 
The derivative of $F_{X|Z}(w|z)$ w.r.t.\ $z$ is  
\begin{align*}
\frac{\int_{\infx}^{w}\frac{df_{Z|X}(z|x)}{dz}dF_{X}(x) \int_{\infx}^{\supx}f_{Z|X}(z|y)dF_{X}(y) -\int_{\infx}^{\supx}\frac{df_{Z|X}(z|y)}{dz}dF_{X}(y) \int_{\infx}^{w}f_{Z|X}(z|x)dF_{X}(x)}{\left[\int_{\infx}^{\supx}f_{Z|X}(z|y)dF_{X}(y)\right]^2}, 
\end{align*} 
nonnegative iff 
\begin{align}
\label{unbounded}
\frac{\int_{\infx}^{w}\frac{df_{Z|X}(z|x)}{dz}dF_{X}(x)}{\int_{\infx}^{w}f_{Z|X}(z|x)dF_{X}(x)} \geq \frac{\int_{\infx}^{\supx}\frac{df_{Z|X}(z|y)}{dz}dF_{X}(y)}{\int_{\infx}^{\supx}f_{Z|X}(z|y)dF_{X}(y)}. 
\end{align}
At $w=\supx$, equality holds in~(\ref{unbounded}). The LHS of~(\ref{unbounded}) is continuous in $w$ (even if $\frac{df_{Z|X}(z|x)}{dz}$ has jumps), so if the LHS decreases in $w$, then the LHS exceeds the RHS at all $w$. 

If $F_{X}$ has a discontinuity of height $H_x$ at $x$, then the Dirac delta function $\delta_x$ is used to represent the density: $f_{X}(x)=H_x\delta_x$. 
Define $h(z|x):=\frac{d\ln f_{Z|X}(z|x)}{dz}$. 
Because $\frac{df_{Z|X}(z|x)}{dz} =f_{Z|X}(z|x) h(z|x)$, the derivative of the LHS of~(\ref{unbounded}) w.r.t.\ $w$ is 
\begin{align*}
&\frac{f_{Z|X}(z|w)h(z|w)f_{X}(w)\int_{\infx}^{w}f_{Z|X}(z|x)dF_{X}(x) -f_{Z|X}(z|w)f_{X}(w)\int_{\infx}^{w}f_{Z|X}(z|x)h(z|x)dF_{X}(x)}{\left[\int_{\infx}^{w}f_{Z|X}(z|x)dF_{X}(x)\right]^2} 
\\&=\frac{f_{Z|X}(z|w)f_{X}(w)\int_{\infx}^{w}[h(z|w)-h(z|x)]f_{Z|X}(z|x)dF_{X}(x) }{\left[\int_{\infx}^{w}f_{Z|X}(z|x)dF_{X}(x)\right]^2} \leq 0, 
\end{align*}
because $h(z|\cdot)$ is decreasing. 
Therefore $\frac{dF_{X|Z}(w|z)}{dz}\geq 0$, which implies that for any $z_1<z_2$ close enough to $z$, $F_{X|Z}(w|z_1)$ FOSDs $ F_{X|Z}(w|z_2)$. 
If $h(z|\cdot)$ strictly decreases, then $\frac{dF_{X|Z}(w|z)}{dz}> 0$. 
\end{proof}


The assumption $X\leq \supx <\infty$ in Theorem~\ref{thm:unbounded} is for comparability to CH. Switching the signs of $X$, $Z$, $\infx$, $\supx$ and $w$ yields the following corollary. 
\begin{cor}
\label{cor:unbounded}
If $X\geq \infx>-\infty$ and $\frac{d\ln f_{Z|X}(z|x)}{dz}$ decreases in $x$, then $\frac{dF_{X|Z}(w|z)}{dz} \geq 0$ for any $w\geq \infx$, with strict inequality if $\frac{d\ln f_{Z|X}(z|\cdot)}{dz}$ is strictly decreasing and $w\in(\infx,\supx)$.  
\end{cor}

If the noise $\tilde{\epsilon}$ is independent of $X$, then Theorem~\ref{thm:unbounded} may be restated as follows. 
\begin{cor}
\label{cor:independent}
If $\tilde{\epsilon}$ is independent of $X$, $\supx<\infty$ and there exists $\peake$ s.t.\  $\frac{f_{\tilde{\epsilon}}'(\epsilon)}{f_{\tilde{\epsilon}}(\epsilon)}$ increases for $\epsilon\geq \peake$, then $F_{X|Z}(w|z_1) \leq F_{X|Z}(w|z_2)$ for any $w<\supx$ and $z_2 >z_1 \geq\supx +\peake$, and if $\frac{f_{\tilde{\epsilon}}'(\epsilon)}{f_{\tilde{\epsilon}}(\epsilon)}$ strictly increases, then $F_{X|Z}(w|z_1) < F_{X|Z}(w|z_2)$. 
\end{cor}
\begin{proof}
Define $\hat{\epsilon}=\tilde{\epsilon}-\peake$, $\hat{Z}=Z-\peake$ $f_{\hat{Z}|X}(\hat{z}|x)=f_{\hat{\epsilon}}(\hat{z}-x)$ and apply Theorem~\ref{thm:unbounded}. 

If $\frac{d\ln f_{\hat{\epsilon}}(\hat{z}-x)}{dz}$ decreases in $x$ at one $\hat{z}\geq\supx$, then it decreases in $x$ at any $\hat{z}\geq\supx$. Therefore $\frac{dF_{X|\hat{Z}}(w|\hat{z})}{d\hat{z}} =\frac{dF_{X|Z}(w|z-\peake)}{dz}\geq 0$ for any $z\geq \supx+\peake$, so $F_{X|Z}(w|z_1) \leq F_{X|Z}(w|z_2)$ for any $z_2 >z_1 \geq\supx +\peake$. If $\frac{d\ln f_{\hat{\epsilon}}(\hat{z}-x)}{dz}$ strictly decreases, then $\frac{dF_{X|\hat{Z}}(w|\hat{z})}{d\hat{z}}> 0$. 
\end{proof}

The noise in Corollary~\ref{cor:independent} has to be unbounded above if $X$ is unbounded below, as Corollary~\ref{cor:ruleout} shows. 
The noise may be symmetric, single-peaked and mean-zero, but these conditions are not necessary.\footnote{ 
To define a signal $S$ such that the event $S\geq z$ contains the same information as $Z=z$ (as in Theorem~\ref{thm:extremescreeningold}), the pdf of $Z$ needs to be decreasing, with a maximum value $1$. In this case, $f_{\tilde{\epsilon}}$ cannot be symmetric. 
}
If $X$ and $\tilde{\epsilon}$ are bounded, then the result of Corollary~\ref{cor:independent} holds for $\supe-\infe>\supx-\infx$ and $z_1,z_2\in[\supx +\peake,\infx+\supe]$. 
The result does not extend to a FOSD reversal for all signal pairs $z_1<z_2$ even if $\peake=\infe$, because $z_1<\supx+\infe$ rules out $X\in(z_1-\infe,\supx]$ that $z_2\geq \supx$ permits. 

The assumptions of Corollary~\ref{cor:independent} are satisfied by exponential noise $f_{\tilde{\epsilon}}(\epsilon) =e^{-\lambda\epsilon}$ for $\tilde{\epsilon},\lambda>0$, as well as any other noise pdf thicker-tailed than exponential. 
Such noise distributions are relevant in practice, for example in detecting tax evasion. 
Suppose the log of income or wealth is $-X\geq -\supx =0$, with $f_{X}(x)=e^x$, 
because empirically, income and wealth follow Pareto distributions. The log \emph{declared} income is $-Z\leq -X$, so $-\tilde{\epsilon}=X-Z\leq 0$ is the log of the fraction of income that is declared. 
The tax authority observes a specific taxpayer's log declared income $z$ and Bayesian updates its belief about the true income. 
Conditional on observable characteristics of the taxpayer, the updated distribution over income levels given a lower declaration FOSDs the distribution given a higher reported income. 

Another interpretation is that $-X\geq -\supx \in(-\infty,0]$ is wealth, $\tilde{\epsilon}\geq 0$ the amount by which wealth is underreported, and $-Z=-X-\tilde{\epsilon}$ the reported wealth, with $f_{X}$ and $f_{\tilde{\epsilon}}(\epsilon)$ Pareto distributions. 

Empirically, the conditional expectation of true wealth or income decreases over a range of declared wealth or income levels.  \cite{alstadsaeter+2019}\footnote{\cite{alstadsaeter+2019} Online appendix Table J.1, Excel file AJZAppendixH sheets DistribEvasion, ToStata, and file AJZAppendixG sheet Discloser(NOR1). 
} 
show that in the bottom deciles of declared wealth, a lower declaration is associated with a greater fraction of tax evaders and larger evasion amounts for income and wealth. Total wealth is larger for the first decile of disclosed wealth than for the second. 
In \cite{artavanis+2016}, the three bottom income bins exhibit a negative correlation between the income declared on the tax return and the income the authors infer from the credit extended by a bank. 
\cite{fagereng+2019} find that the fraction of income saved decreases in declared wealth for households declaring negative wealth, but increases for positive-wealth households. The median capital gains rate decreases fastest for the lowest wealth percentiles. Thus households whose net wealth is more negative seem to obtain higher returns on assets and save more, which suggests they under-report assets to a greater extent. 
\cite{johannesen+2018} Figure A.1 shows that among the US persons in the Offshore Voluntary Disclosure program, the ratio of the disclosed account value to the capital income in the previous year is decreasing and convex in their capital income percentile up to the 90th percentile. In other words, those who previously declared less capital income have disproportionately larger offshore accounts. 

\cite{alstadsaeter+2019} find that conditional on hiding wealth, the fraction hidden is approximately 30\% at all wealth levels. The fraction of income undeclared is about 40\% conditional on under-reporting, independently of wealth. The probability of hiding wealth rises with wealth, which may be modelled as follows. 
Denote the probability of evading taxes at (negative) wealth $X=x$ by $p(x)>0$, the pdf of the fraction of income declared conditional on evading by $g$, and the Dirac delta function at $\alpha$ by $\delta_{\alpha}$. 
Define the noise pdf as $f_{\tilde{\epsilon}|X}(\epsilon|x) =p(x)\delta_{0} +(1-p(x))g(\epsilon)$. 
Then
\begin{align*}
F_{X|Z}(w|z) 
=\frac{\text{\textbf{1}}_{z\leq w}f_{X}(z)p(z) +\int_{-\infty}^{w}f_{X}(x)(1-p(x))g(z-x)dx}{\text{\textbf{1}}_{z\leq 0}f_{X}(z)p(z) +\int_{-\infty}^{0}f_{X}(y)(1-p(y))g(z-y)dy}, 
\end{align*} 
which implies that for any continuous $g$, there exists $\overline{z}>0$ such that for any $z\in(0,\overline{z})$ and $w<0$, 
\begin{align*}
F_{X|Z}(w|0) 
=\frac{\int_{-\infty}^{w}f_{X}(x)(1-p(x))g(-x)dx}{f_{X}(0)p(0) +\int_{-\infty}^{0}f_{X}(y)(1-p(y))g(-y)dy}
<\frac{\int_{-\infty}^{w}f_{X}(x)(1-p(x))g(z-x)dx}{\int_{-\infty}^{0}f_{X}(y)(1-p(y))g(z-y)dy} =F_{X|Z}(w|z). 
\end{align*}
If $g(\epsilon)=e^{-\lambda \epsilon}$ with $\lambda>0$, or $g$ is thicker-tailed than exponential, then $\overline{z}=\infty$. 

\section{Discussion}
\label{sec:discussion}

The connection established in this note between CH and LL allows using the stronger result in the former to study the diverse applications in the latter. For example, a journal editor screening papers based on referee reports may obtain not just a better expected quality of papers by reducing the acceptance cutoff, but a FOSD-improved distribution of quality. 
A bank choosing borrowers or an investor picking projects based on quantitative criteria, such as credit ratings, may shift the return distribution up (in the FOSD sense) by relaxing the criteria. 
The optimal strategy of accepting applicants to educational institutions or conducting an affirmative action policy may not be monotone in people's observable characteristics. Switching to a better strategy may improve outcomes at every point of the distribution of characteristics. 
Similarly, in an auction, 
requiring the opening bid to be in a disconnected set (as opposed to above a cutoff) may FOSD-improve the revenue distribution. 

The extension of the results to unbounded distributions, in particular exponential and Pareto, opens up additional applications: tax evasion and measurement errors in the wealth and income distributions. For example, the asset distribution of people reporting a lower wealth level may FOSD the distribution of those making a higher report. The extended results hold for a wide range of signal pairs, thus the FOSD reversal is a robust phenomenon.

\appendix
\section{Notes on other results in LL}

LL Theorem 2 implicitly assumes that the utility functions are strictly increasing and the FOSD ranking of the conditional distributions $F_{X|Z\in[z_1,z_2]}$ 
is strict. If weak FOSD is allowed, then the signal $Z$ could be uninformative over an interval $[z_1,z_2]$, in which case the expected value $\mathbb{E}[X|Z=z] =:x_E$ of the random variable of interest is the same for any $z\in[z_1,z_2]$. Then for any utility function $u$ such that $u(x_E)=0$, there is a continuum of optimal strategies, because changing the acceptance probability on any subset of $[z_1,z_2]$ results in a different optimal strategy. Most of these optimal strategies are not cutoff strategies. 
If $u$ is permitted to be zero on an interval $[x_1,x_2]$ of realizations of $X$, then there is similarly a continuum of optimal non-cutoff strategies that differ by the acceptance probability on $[x_1,x_2]$. 

The statement of Claim 3 in LL only assumes $X$ is continuous and $\tilde{\epsilon}=Z-X$ is discrete with finite support $n_1<n_2<\cdots<n_k$, but the proof implicitly assumes that $X$ is bounded above: $X\leq \overline{X}<\infty$. The proof uses the sets $A:=\left(\overline{X}+n_{k-1},\overline{X}+n_{k-1}+\varepsilon\right)$ and $B:=\left(\overline{X}-n_{k-1}-\varepsilon,\overline{X}-n_{k-1}\right)$ to define the support $supp(X|_A):=supp F_{X|Z\in A}$ and similarly for $B$. If $\overline{X}=\infty$, then $supp(X|_A)=supp(X|_B)=\emptyset$ and the probabilities $\Pr\left(X>\overline{X}-\varepsilon|A\right)$ and $ \Pr\left(X>\overline{X}-\varepsilon|B\right)$ are undefined.

\bibliographystyle{chicago} 
\bibliography{teooriaPaberid} 

\end{document}